\documentclass[12pt, letter]{article}
\usepackage{natbib}
\usepackage{latexsym}
\usepackage{graphics}
\usepackage{amsfonts}
\usepackage{enumerate}
\usepackage{graphicx}
\usepackage{epsf,boxedminipage,lscape}
\usepackage{epstopdf}
\usepackage{epsfig}
\usepackage{verbatim}
\usepackage{subfig}
\usepackage{caption}
\usepackage{footnote}
\usepackage{slashbox}
\usepackage{lscape}
\usepackage{hhline}
\usepackage{array}
\usepackage{graphicx}
\usepackage{setspace}
\usepackage{amssymb}
\usepackage{amsmath}
\usepackage{amsfonts}
\usepackage{times}
\usepackage{color}
\usepackage{amsthm}
\usepackage{indentfirst}
\usepackage{epsfig}
\usepackage{lscape,ctable,rotating}
\usepackage{enumitem}
\setlist{nolistsep}
\usepackage[left=1in,top=1in,right=1in,foot=1in]{geometry}

\usepackage{amsmath}
\DeclareMathOperator*{\argmax}{arg\,max}
\DeclareMathOperator*{\argmin}{arg\,min}

\newtheorem{definition}{Definition}[section]
\newtheorem{lemma}[definition]{Lemma}

\newtheorem{exam}[definition]{Example}
\newtheorem{proposition}{Proposition}[section]
\newtheorem{rem}{Remark}[section]

\newenvironment{example}{\begin{exam} \rm{ }}{\qed \end{exam}}

\newcommand{\R}{\mathbb{R}}

\newcommand{\var}{\textrm{\rm var}}
\newcommand{\cov}{\textrm{\rm cov}}

\newcommand{\re}{ \textrm{\rm Re} }
\newcommand{\nop}[1]{{#1}}

\newcommand{\levy}{{L\'evy}}

\newcommand{\Tau}{{\mathcal{T}}}

\newcommand{\disteq}{\stackrel{\mathrm{d}}{=}}

\newcommand{\NTS}{\textrm{\rm NTS}}
\newcommand{\stdNTS}{\textrm{\rm stdNTS}}
\newcommand{\tr}{{\texttt{T}}}
\newcommand{\diag}{\textrm{\rm diag}}

\newcommand{\VaR}{\textrm{\rm VaR}}
\newcommand{\CVaR}{\textrm{\rm CVaR}}

\title{Portfolio Optimization on the Dispersion Risk and the Asymmetric Tail Risk}
\author{Young Shin Kim\footnote{College of Business, Stony Brook University, New York, USA (aaron.kim@stonybrook.edu). The author gratefully acknowledges the support of GlimmAnalytics LLC and Juro Instruments Co., Ltd.
The author is grateful to Minseob Kim, who reviewed this paper and corrected editorial errors. Also, all remaining errors are entirely my own.} }
\providecommand{\keywords}[1]{\textbf{\textit{Key words:}} #1}

\begin{document}
\maketitle
 
\begin{abstract}
In this paper, we propose a market model with returns assumed to follow a multivariate normal tempered stable distribution defined by a mixture of the multivariate normal distribution and the tempered stable subordinator. This distribution is able to capture two
stylized facts: fat-tails and asymmetry, that have been empirically observed for asset return distributions. On the new market model, we discuss a new portfolio optimization method, which is an extension of Markowitz’s mean-variance optimization. The new optimization method considers not only reward and dispersion but also asymmetry. The efficient frontier is extended to a curved surface on three-dimensional space of reward, dispersion, and asymmetry. We also propose a new performance measure which is an extension of the Sharpe Ratio. Moreover, we derive closed-form solutions for two important measures used by portfolio managers in portfolio construction: the marginal Value-at-Risk (VaR) and the marginal Conditional VaR (CVaR). We illustrate the proposed model using stocks comprising the Dow Jones Industrial Average. First, perform the new portfolio optimization and then demonstrating how the marginal VaR and marginal CVaR can be used for portfolio optimization under the model. Based on the empirical evidence presented in this paper, our framework
offers realistic portfolio optimization and tractable methods for portfolio risk management.\\
\end{abstract}

\baselineskip=24pt

\doublespacing
\section{Introduction}

It admits no doubt that the mean-variance model formulated by Harry Markowitz (\citeyear{Markowitz:1952}) is a major contribution to the portfolio theory in finance. Although some assumptions of the model were challenged theoretically and empirically, it is impossible to exaggerate the importance of Markowitz's portfolio optimization theory. The optimization model has been applied to  asset allocation, portfolio selection, asset-liability management, and risk management. There are many alternative models by relaxation of the assumption of the model. In this paper, we discuss two assumptions of the model and how we can improve them.

The first is the assumption that asset returns follow a Gaussian distribution, which has been dominateing financial theories. However, the Gaussian assumption was empirically rejected in literature including \cite{Mandelbrot:1963a,Mandelbrot:1963b} and \cite{Fama:1963}, and empirical evidence says that asset returns exhibit fat tails and asymmetry. Therefore, non-Gaussian models, which can describe the stylized facts about asset return data, have been suggested for the underlying model of the portfolio optimization. The subordinated Gaussian distribution is popularly used to construct a multivariate market model with fat-tails and asymmetry. This distribution of the model is defined by taking the multivariate normal distribution and changing the variance to a strictly positive random vector, which is referred to as the subordinator. Examples of the distribution used in the portfolio theory are $\alpha$-stable subordinated Gaussian distribution (\cite{RachevMittnik:2000}), the inverse Gaussian subordinated Gaussian distribution (\cite{Oigaard:2005}, \cite{Aas:2006}, \cite{EberleinMadan:2010}), the inverse Gamma subordinated Gaussian distribution (\cite{StoyanovRachevFabozzi:2013}) and the tempered stable subordinated Gaussian distribution (\cite{BarndorffNielsenShephard:2001} and \cite{BarndorffNielsenLevendorskii:2001}).

The second assumption is the use of the portfolio variance as a risk measure. 
Since asset return distribution does not follow Gaussian distribution but exhibits fat-tails and asymmetry, a risk measure would better be able to assess  not only dispersion but also asymmetry. Portfolio optimization with asymmetry risk measure has been discussed in \cite{King:1993} and \cite{Dahlquist_et_al:2017}. Since the asymmetric risk indirectly measured by the coherent risk measures such as the conditional value at risk (CVaR) by \cite{Pflug:2000} and \cite{RockafellarUryasev:2000,RockafellarUryasev:2002}, the portfolio optimization with coherent risk measure has been studied in \cite{RachevStoyanovFabozzi:2007}, \cite{Mansini_et_al:2007}, \cite{RachevKimBianchiFabozzi:2011a}, and \cite{Kim_et_al:2012}.

In this paper, we propose a non-Gaussian market model that returns are assumed to follow the normal tempered stable (NTS) distribution. The NTS distribution is the tempered stable subordinated Gaussian distribution. It is asymmetric and has exponential tails that are fatter than Gaussian tails and thinner than the power tails of $\alpha$-stable distributions. For that reason, it can describe the asymmetric and fat-tail properties of the stock return distribution. Since it has exponential tails, it has finite exponential moments and finite integer moments for all orders such as mean, variance, skewness, and kurtosis. By standardization, we obtain a NTS distribution having zero mean and unit variance. That can be used as the innovation distribution of a time series model including the ARMA-GARCH model. Since the distribution is infinitely divisible, we can generate a continuous-time \levy~process on the NTS distribution. For this good properties, the distribution was popularly used in finance; portfolio optimization(\cite{EberleinMadan:2010}, \cite{Kim_et_al:2012}, \cite{Anad_et_al:2016}), risk management (\cite{Kim_et_al:2009b:AVaR}, \cite{Anad_et_al:2017}, \cite{KurosakiKim:2018}), option pricing (\cite{Boyarchenko_Levendorskii:2000}, \cite{RachevKimBianchiFabozzi:2011a}, \cite{EberleinGlau:2014}, \cite{KIM2015512}), term structure of interest rate model (\cite{Eberlein_Ozkan:2005}), and credit risk management(\cite{Eberlein_et_al:2012}, \cite{KimKim:2018}).

An important measure derived from portfolio optimization is the marginal risk contribution. The marginal risk contribution help managers  to make portfolio rebalancing decisions. This risk measure is the rate of change in risk, whether it is variance, Value-at-Risk(VaR), or CVaR, with respect to a small percentage change in the size of a portfolio allocation weight. Mathematically it is defined by the first derivative of the risk measure with respect to the marginal weight. Because of the importance of this measure in portfolio decisions, a closed-form solution for this measure is needed. The general form of marginal risk contributions for the VaR and CVaR are provided in \cite{GourierouxaEtAl:2000}. Moreover, the closed-form of the marginal risk contributions for VaR and CVaR are discussed under the skewed-$t$ distributed market model in \cite{StoyanovRachevFabozzi:2013}, under the NTS market model in \cite{Kim_et_al:2012}, and under the Generalized Hyperbolic distribution model in \cite{ShiKim:2015}.

The contributions of this paper are as follows. First, we construct a market model using the multivariate NTS distributed portfolio returns. Different from the NTS market model in \cite{Kim_et_al:2012}, the market model in this paper is constructed by only the standard NTS distribution which is a subclass of NTS distribution having zero mean and unit variance. Using the new NTS market model, we discuss the portfolio optimization theory considering dispersion risk and asymmetry risk. As the Markowitz model, the standard deviation is used for the dispersion risk measure. In addition to the dispersion risk measure, the asymmetric tail risk measure is proposed in order to capture the asymmetry risk in portfolio optimization. It is a weighted mean of asymmetry parameters of the standard NTS distribution. Using those two risk measures, we find an efficient frontier surface which is an extension of Markowitz's efficient frontier \footnote{The similar study has been presented in \cite{ShiKim:2015} under the Generalized Hyperbolic distribution model.}. 
Moreover, a new performance measure of a portfolio on an efficient frontier surface is presented. The new performance measure is an extension of the Sharpe ratio (\cite{Sharpe:1966,Sharpe:1994}). Finally, we provide closed-form solutions of the marginal risk contribution for VaR and CVaR under the NTS market model. The marginal VaR and marginal CVaR formula in this paper are simpler than the solutions presented in \cite{Kim_et_al:2012}, and hence we can discuss the iterative risk budgeting which was not discussed in \cite{Kim_et_al:2012}. 

Empirical illustrations are provided for each topic in the paper with performance tests. Data used in the empirical illustrations are historical daily returns of major 30 stocks in the U.S. market. We draw the efficient frontier surface based on the estimated parameters of the NTS market model. The performance measure maximization strategy is also exhibited and it is compared to the traditional Sharpe ratio maximization strategy by backtesting. We calculate Marginal VaR \& CVaR and perform the risk budgeting using calculated marginal VaR \& CVaR. 

The remainder of this paper is organized as follows. 
The NTS market model is presented in Section 2. The portfolio optimization with dispersion and asymmetry risk measures is discussed in Section 3 under the NTS market model. The new performance measure of the portfolio is presented in Section 4. We provide closed-form solutions for the marginal VaR and CVaR under the NTS market model in Section 5, where we also discuss portfolio budgeting
using the marginal VaR and the marginal CVaR. Finally, Section 6 concludes. Proofs and mathematical details are presented in the appendix.

\section{\label{sec:NTS}Normal Tempered Stable Market Model}
Let $N$ be a finite positive integer and $X=(X_1, X_2, \cdots, X_N)^{\tr}$ be a multivariate random variable given by
\[
X = \beta(\Tau-1) + \textup{diag}(\gamma) \varepsilon \sqrt{\Tau} ,
\]
where 
\begin{itemize}
\item $\Tau$ is the tempered stable subordinator\footnote{The tempered subordinator is defined by the characteristic function \eqref{eq:ChF.TSsubordProcess} in the Appendix.} with parameters $(\alpha,\theta)$, and is independent of $\varepsilon_n$ for all $n=1,2,\cdots, N$.
\item $\beta = (\beta_1, \beta_2, \cdots, \beta_N)^{\tr}\in\R^N$ with $|\beta_n|<\sqrt\frac{2\theta}{2-\alpha}$ for all $n\in\{1,2,\cdots, N\}$.
\item $\gamma = (\gamma_1, \gamma_2, \cdots, \gamma_N)^{\tr}\in\R_+^N$ with $\gamma_n = \sqrt{1-\beta_n^2\left(\frac{2-\alpha}{2\theta}\right) }$  for all $n\in\{1,2,\cdots, N\}$ and $\R_+=[0,\infty)$.
\item $\varepsilon = (\varepsilon_1, \varepsilon_2, \cdots, \varepsilon_N)^{\tr}$ is $N$-dimensional standard normal distribution with a covariance matrix $\Sigma$. That is, $\varepsilon_n\sim \Phi(0,1)$ for $n\in\{1,2,\cdots, N\}$ and $(k,l)$-th element of $\Sigma$ is given by $\rho_{k,l}=\cov(\varepsilon_k,\varepsilon_l)$ for $k,l\in\{1,2,\cdots,N\}$. Note that $\rho_{k,k}=1$.
\end{itemize}
In this case, $X$ is referred to as the \textit{$N$-dimensional standard NTS random variable} with parameters $(\alpha$, $\theta$, $\beta$, $\Sigma)$ and we denote it by $X\sim \textup{stdNTS}_N(\alpha$, $\theta$, $\beta$, $\Sigma)$ (See more details in Appendix.).

Consider a portfolio having $N$ assets. The return of the assets in the portfolio is given by a random vector $R=(R_1$, $R_2$, $\cdots$, $R_N)^\tr$. We suppose that the return $R$ follows 
\begin{equation} \label{eq:NTS Market}
R = \mu + \diag(\sigma) X 
\end{equation}
where  $\mu = (\mu_1, \mu_2, \cdots, \mu_N)^\tr\in\R^N$, $\sigma = (\sigma_1, \sigma_2, \cdots, \sigma_N)^\tr\in\R_+^N$ and $X\sim \stdNTS_N(\alpha, \theta, \beta, \Sigma)$. Then we have $E[R_n]=\mu_n$ and $\var(R_n) = \sigma_n^2$ for all $n\in\{1,2, \cdots, N\}$. This market model is referred to as the \textit{NTS market model}.
Let $w=(w_1, w_2, \cdots, w_N)^\tr\in I^N$ with $I=[0,1]$ be the capital allocation weight vector\footnote{In this paper, we consider the long only portfolio.}. Then the portfolio return for $w$ is equal to $R_P(w)=w^\tr R$. The distribution of $R_P(w)$ is presented in the following proposition whose proof is in Appendix.
\begin{proposition}\label{pro:mu+sigma stdNTS}
Let $\mu \in\R^N$,
$\sigma \in \R_+^N$, and $X \sim \textup{stdNTS}_N(\alpha, \theta, \beta, \Sigma)$. Suppose a $N$-dimensional random variable $R$ is given by \eqref{eq:NTS Market} and $w\in I^N$ with $I=[0,1]$. Then
\begin{equation} \label{eq:NTS Portfolio Return}
R_P(w) \disteq \bar\mu(w) + \bar\sigma(w) \Xi ~~~\text{ for }~~~ \Xi\sim \textup{stdNTS}_1(\alpha, \theta, \bar\beta(w), 1),
\end{equation}
where 
\[
\bar\mu(w) = w^\tr\mu,~~~ \bar\sigma(w) = \sqrt{w^\tr \Sigma_R w},
~~~\bar\beta(w)=\frac{w^\tr\diag(\sigma)\beta}{\bar\sigma(w)}
\]
and $\Sigma_R$ is the covariance matrix of $R$.
\end{proposition}
According to Proposition \ref{pro:mu+sigma stdNTS}, we need only $\Sigma_R$, which is covariance matrix of $R$, and we do not need to know $\Sigma$, which is the covariance matrix for $\epsilon$, when we study the portfolio return of the NTS market model.
\\
\begin{table}[!ht]
\centering
\begin{tabular}{cc|cc}
\hline
Company & Symbol & Company & Symbol \\
\hline
\hline
3M	&	 MMM	&
American Express	&	 AXP	\\
Apple Inc.	&	AAPL	 &
Boeing	&	 BA	\\
Caterpillar Inc.	&	 CAT	&
Chevron Corporation	&	 CVX	\\
Cisco Systems	&	CSCO	&
The Coca-Cola Company	&	 KO	\\
DuPont de Nemours Inc.	 &	 DD	&
Exxon Mobil	&	 XOM	\\
Goldman Sachs	&	 GS	&
The Home Depot	&	 HD	\\
IBM	&	 IBM	&
Intel	&	INTC	\\
Johnson \& Johnson	&	 JNJ	&
JPMorgan Chase	&	 JPM	\\
McDonald's	&	 MCD	&
Merck \& Co.	&	 MRK	\\
Microsoft	&	MSFT	&
Nike	 &	 NKE	\\
Pfizer	&	 PFE	&
Procter \& Gamble	&	 PG	\\
Raytheon Technologies&	 RTX	&
The Travelers Companies	&	 TRV	\\
United Health Group	&	 UNH	&
Verizon	&	 VZ	\\
Visa Inc.	&	 V	&
Walmart	&	 WMT	\\
Walgreens Boots Alliance	&	WBA	&
The Walt Disney Company	&	 DIS	\\
\hline
\end{tabular}
\caption{\label{table:DJIA Members}Companies and symbols of 30 Stocks. They are selected based on the components for DJIA index, since April 6, 2020, but Dow Inc.(DOW) in the components is replaced by DuPont de Nemours Inc.(DD).}
\end{table}
~\\
\textit{Empirical Illustration}

We fit the NTS market model to 30 major stocks in the U.S. stock market. The 30 stocks are selected based on the components for Dow Johns Industrial Average (DJIA) index, since April 6, 2020. However, Dow Inc.(DOW) in the components is replaced by DuPont de Nemours Inc.(DD), since DOW does not have enough history. Table \ref{table:DJIA Members} exhibits the names of companies and symbols of those 30 Stocks.
For each stock, we calculate the sample mean and sample standard deviation of daily log-returns, and then fit the stdNTS parameters using extract standardized residual. In this parameter fitting, the curve fit method is used between the cumulative distribution function (CDF) of the stdNTS distribution\footnote{CDF of stdNTS distribution is obtained by the fast Fourier transform method by \cite{Gil-Pelaez:1951}. } and the empirical CDF obtained by the kernel density estimation. In order to find $\alpha$ and $\theta$, we use the DJIA index data and find other parameters as the following two-step method:
\begin{itemize}
\item[] \textbf{Step 1} Find $(\alpha_{DJ}, \theta_{DJ}, \beta_{DJ})$ using the curve fit method between the empirical CDF and stdNTS CDF for the standardized return data of DJIA index.
\item[] \textbf{Step 2} Fix $\alpha = \alpha_{DJ}$ and $\theta = \theta_{DJ}$. Find $\beta_n$ using the curve fit method again for each $n$-th stock returns $n\in\{1,2,\cdots, 30\}$ with fixed $\alpha$ and $\theta$.
\end{itemize}

\begin{table}
\centering
\begin{tabular}{cccccc}
\hline
Symbol & $\mu_n (\%)$ & $\sigma_n (\%)$ & $\beta_n (\%)$  & KS & $p$-value$(\%)$\\
\hline
\hline
AAPL & $0.13 $ & $ 1.56 $ & $ -2.55 $ & $ 0.022 $ & $ 83.74 $ \\ 
  AXP & $0.073 $ & $ 1.18 $ & $ -3.90 $ & $ 0.036 $ & $ 29.01 $ \\ 
  BA & $0.11 $ & $ 1.68 $ & $ 0.50 $ & $ 0.038 $ & $ 22.08 $ \\ 
  CAT & $0.071 $ & $ 1.71 $ & $ -2.93 $ & $ 0.035 $ & $ 29.77 $ \\ 
  CSCO & $0.071 $ & $ 1.44 $ & $ -8.56 $ & $ 0.015 $ & $ 99.57 $ \\ 
  CVX & $0.018 $ & $ 1.22 $ & $ -4.65 $ & $ 0.032 $ & $ 41.73 $ \\ 
  DD & $-0.020 $ & $ 1.70 $ & $ 0.32 $ & $ 0.034 $ & $ 32.84 $ \\ 
  DIS & $0.046 $ & $ 1.24 $ & $ 2.00 $ & $ 0.029 $ & $ 55.60 $ \\ 
  GS & $-0.00044 $ & $ 1.46 $ & $ -1.01 $ & $ 0.043 $ & $ 11.57 $ \\ 
  HD & $0.074 $ & $ 1.16 $ & $ -7.01 $ & $ 0.044 $ & $ 10.41 $ \\ 
  IBM & $-0.013 $ & $ 1.30 $ & $ -7.31 $ & $ 0.012 $ & $ 99.99 $ \\ 
  INTC & $0.075 $ & $ 1.70 $ & $ -1.19 $ & $ 0.024 $ & $ 74.81 $ \\ 
  JNJ & $0.041 $ & $ 1.09 $ & $ -7.53 $ & $ 0.020 $ & $ 91.31 $ \\ 
  JPM & $0.072 $ & $ 1.21 $ & $ 3.85 $ & $ 0.042 $ & $ 13.93 $ \\ 
  KO & $0.050 $ & $ 0.89 $ & $ -2.67 $ & $ 0.032 $ & $ 41.47 $ \\ 
  MCD & $0.076 $ & $ 1.04 $ & $ -3.17 $ & $ 0.024 $ & $ 77.08 $ \\ 
  MMM & $0.0093 $ & $ 1.36 $ & $ -10.55 $ & $ 0.0084 $ & $ 99.999 $ \\ 
  MRK & $0.067 $ & $ 1.14 $ & $ -0.68 $ & $ 0.037 $ & $ 24.66 $ \\ 
  MSFT & $0.13 $ & $ 1.36 $ & $ -4.47 $ & $ 0.027 $ & $ 65.24 $ \\ 
  NKE & $0.093 $ & $ 1.52 $ & $ 0.44 $ & $ 0.030 $ & $ 50.61 $ \\ 
  PFE & $0.037 $ & $ 1.08 $ & $ -7.19 $ & $ 0.032 $ & $ 41.80 $ \\ 
  PG & $0.064 $ & $ 1.00 $ & $ -0.26 $ & $ 0.025 $ & $ 71.99 $ \\ 
  RTX & $0.050 $ & $ 1.22 $ & $ -7.41 $ & $ 0.024 $ & $ 77.15 $ \\ 
  TRV & $0.025 $ & $ 1.07 $ & $ -10.52 $ & $ 0.036 $ & $ 26.68 $ \\ 
  UNH & $0.086 $ & $ 1.35 $ & $ -1.41 $ & $ 0.029 $ & $ 52.23 $ \\ 
  V & $0.12 $ & $ 1.22 $ & $ -10.07 $ & $ 0.023 $ & $ 82.92 $ \\ 
  VZ & $0.033 $ & $ 1.12 $ & $ -4.12 $ & $ 0.040 $ & $ 16.65 $ \\ 
  WBA & $-0.036 $ & $ 1.58 $ & $ -12.62 $ & $ 0.021 $ & $ 88.22 $ \\ 
  WMT & $0.083 $ & $ 1.21 $ & $ -1.64 $ & $ 0.019 $ & $ 93.86 $ \\ 
  XOM & $-0.019 $ & $ 1.12 $ & $ -4.16 $ & $ 0.051 $ & $ 4.14 $ \\ 
\hline
\multicolumn{6}{c}{$\alpha=0.9766$ and $\theta = 0.2253$}
\\
\hline
\end{tabular}
\caption{\label{table:ParamEst}NTS parameter fit using 1-day-returns from 1/1/2017 to 12/31/2019.  }
\end{table}

The estimated parameters of $\mu_n$, $\sigma_n$ and $\beta_n$ are presented in Table \ref{table:ParamEst}. The fixed $\alpha$ and $\theta$ are in the bottom of the table. Those parameters are estimated using 1-day log-returns of the 30 stocks from 1/1/2017 to 12/31/2019.
We can see that only 5 stocks (BA, DD, DIS, JPM, NKE) out of those 30 stocks have positive betas in this table. That means, 25 stocks follow left-skewed return distributions, while only 5 stocks follow positively skewed distributions. The parameter $\alpha$ closes to 1, and $\theta$ is less than 1. That means the 30-dimensional distribution for those 30 stock returns has fat-tails\footnote{More precisely, it has the semi-fat-tails having the exponential decaying tails}. We perform the Kolmogorov-Smirnov (KS) goodness of fit test. KS statistic (KS) values and those $p$-values are presented in the table too. According to the KS $p$-values in the table, the marginal NTS distributions are not rejected at 5\% significant level except XOM. At 4\% significant level, the marginal NTS distribution is not rejected for XOM either. We do not need to estimate the covariance matrix $\Sigma$ for $\varepsilon$, since the covariance matrix $\Sigma_R$ of $R$ is enough to analysis the portfolio return as discussed in Proposition \ref{pro:mu+sigma stdNTS}.

\section{Portfolio Optimization and Efficient Frontier}
In traditional mean-variance portfolio optimization, the investor finds the optimal portfolio which maximizes the reward and minimizes the dispersion risk.
The reward is the expected return of the portfolio and the dispersion risk is the standard deviation (or the variance) of the portfolio. However, the dispersion risk is only one feature of the portfolio risk, while there are many other risks for instance the asymmetric tail risk. If there are two portfolios having same mean and variance, and one of them follows a skewed right distribution and the other follows a skewed left distribution, then the skewed left distributed portfolio is more risky than the skewed right distributed portfolio. We take the NTS market model and measure both the dispersion risk and asymmetric tail risk of a given portfolio.

The classical dispersion risk measure introduced by \cite{Markowitz:1952} is the standard deviation of the portfolio defined as follows:
\[
\textup{Disp.}(w) = \sqrt{w^\tr \Sigma_R w}
\]
where $w$ is the capital allocation weight of the portfolio.
The asymmetric tail risk is related to asymmetric tails of the NTS distribution defined as follows:
\[
\textup{Asym.}(w) = w^\tr \beta.
\]
Since $|\beta_n|<\sqrt\frac{2\theta}{2-\alpha}$ for all $n\in\{1,2,\cdots, N\}$ by the definition of the NTS market model, $|\textup{Asym.}(w)|$ cannot be larger than $\sqrt\frac{2\theta}{2-\alpha}$, in $0\le w_n\le 1$ for all $n\in\{1,2,\cdots, N\}$.
A portfolio having positive $\textup{Asym.}(w)$ follows the right-skewed distribution, and a portfolio having negative $\textup{Asym.}(w)$ follows the left-skewed distribution. If there are two portfolios, A and B, whose capital allocation weight vectors are $w_A$ and $w_B$ respectively, and $\textup{Asym.}(w_A)<\textup{Asym.}(w_B)$, then portfolio A is more risky than portfolio B. That will be discussed again in Example \ref{ex:Asym}.

\begin{figure}[t]
\centering
\includegraphics[width = 12cm]{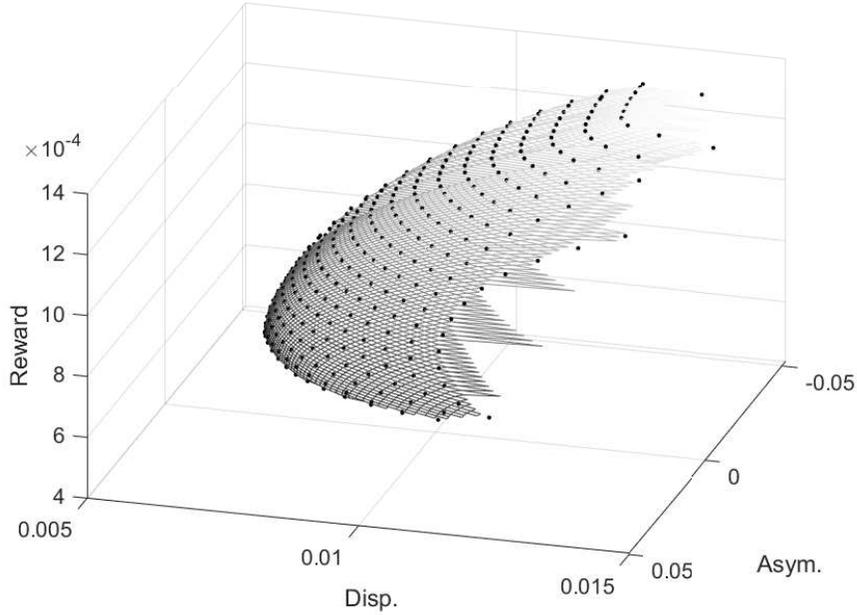}
\caption{\label{fig:Eff3D}Surface of efficient frontiers}
\end{figure}

Based on those two risk measures, we set a nonlinear programming problem for the portfolio optimization as
\begin{align*}
&\hspace{-2cm}\min_{w} \left(\textup{Disp.}(w)\right)\\
\text{subject to}~~~ &\sum_{n=1}^N w_n = 1\\
& w_n\ge 0 \text{ for all } n\in \{1,2,\cdots, N\} \\
& \textup{Asym.}(w) = w^\tr \beta \ge b^* \\
& \textup{Reward}(w) = w^\tr\mu \ge m^* 
\end{align*}
where the benchmark values for Portfolio reward and asymmetric tail risk are
$m^*\in[\min(\mu),\max(\mu)]$ and $b^*\in[\min(\beta),\max(\beta)]$, respectively. 

Using the parameters in Table \ref{table:ParamEst}, we perform the portfolio optimization for 51 points of $b^*$ in $\{b=\min(\beta)+k\cdot (\max(\beta)-\min(\beta))/50\,|\, k=0,1,2,\cdots, 50\}$ and for 51 points of $m^*$ in $\{m=\min(\mu)+k\cdot (\max(\mu)-\min(\mu))/50\,|\, k=0,1,2,\cdots, 50\}$. We finally obtain the efficient frontier surface on the three-dimensional space in Figure \ref{fig:Eff3D}. The dot-points of the surface is $(\textup{Disp.}(w), \textup{Asym.}(w), \textup{Reward}(w))$ for the optimal capital allocation weight $w$. The mash surface of the figure is an interpolation surface of the dot-points. 

\begin{figure}[t]
\centering
\includegraphics[width = 12cm]{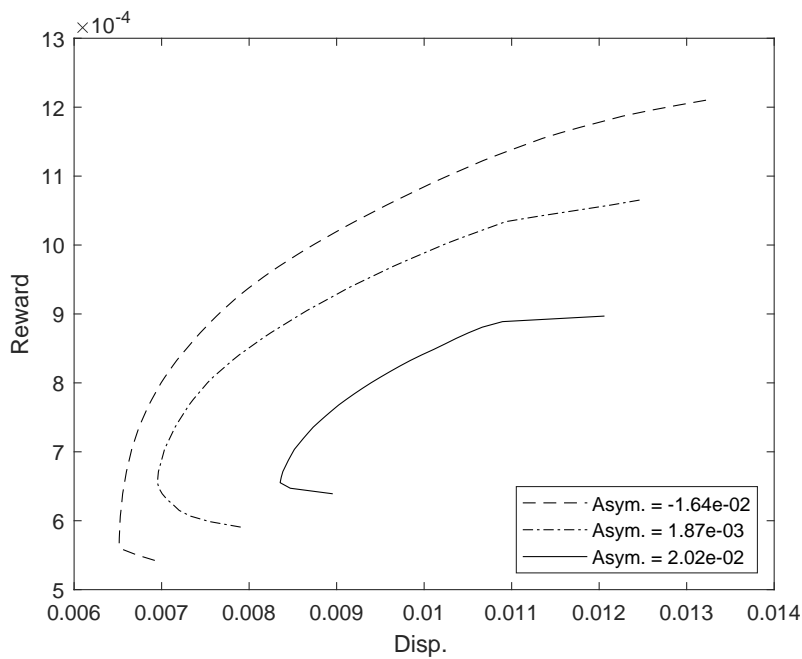}
\caption{\label{fig:Eff2D-Disp}Efficient frontiers between the dispersion risk and the reward. We can observe that the mean-variance efficient frontier is not unique but changes with respect to the asymmetric tail risk.}
\end{figure}

In order to look into the surface, Figure \ref{fig:Eff2D-Disp} provides three examples of efficient frontier curves between the dispersion risk and the reward. The three curves are extracted from Figure \ref{fig:Eff3D}. For instance, we fix $b^*=-1.64\%$, and perform the optimization for $m^*\in [\min(\mu),\max(\mu)]=[-0.036\%, 0.13\%]$, present it to the dash curve of the figure. Using the same method, we draw the dash-dot curve and the solid curve for $b^*=0.187\%$ and $b^*=2.02\%$, respectively. 
Figure \ref{fig:Eff2D-Asym} presents another examples of efficient frontier curves between the asymmetric tail risk and the reward. Those three curves are extracted from Figure \ref{fig:Eff3D}, too. The solid curve is for $(\textup{Asym.}(w), \textup{Reward}(w))$ having $\textup{Disp.}(w)=0.840\%$. Also, the dash-dot and the dashed curves are for $\textup{Disp.}(w)=0.961\%$ and for $\textup{Disp.}(w)=1.09\%$, respectively. According to those two figures, we observe that the classical (mean-variance) efficient frontier is not unique but there are many various form of efficient frontier curves with respect to the asymmetric tail risk. Moreover, we can observe that the reward is decreasing when the asymmetric tail risk is increasing under the fixed dispersion risk.

\begin{figure}[t]
\centering
\includegraphics[width = 12cm]{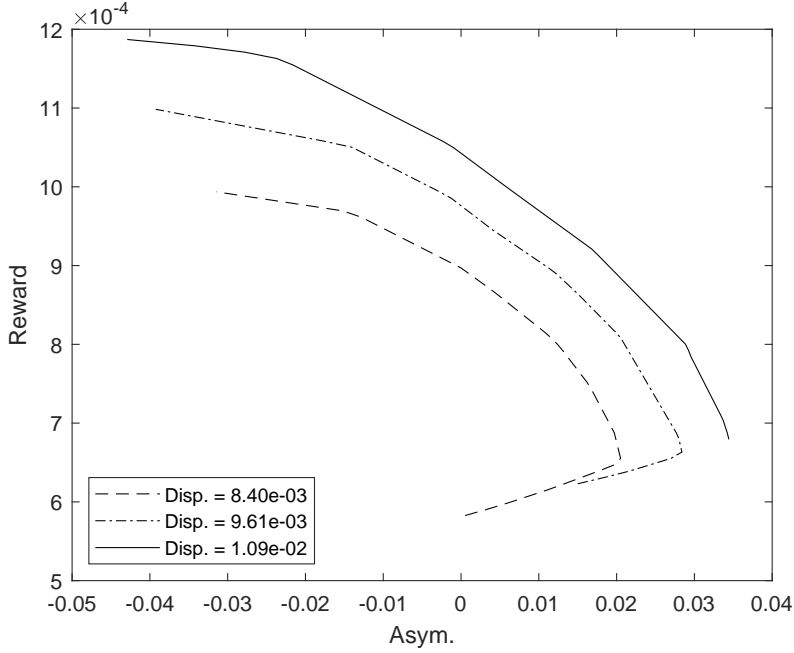}
\caption{\label{fig:Eff2D-Asym}Efficient frontiers between the asymmetric tail risk and the reward. We can observe that the reward is decreasing when the asymmetric tail risk is increasing under the fixed dispersion risk.}
\end{figure}

\section{Performance Measure}
In this section, we discuss performance measure of portfolios. We present a new performance measure on the NTS Market model. The new performance measure is an extension of the Sharpe ratio (\cite{Sharpe:1966,Sharpe:1994}), and it measures performance of portfolios considering not only the dispersion risk but also asymmetric tail risk.

The Sharpe ratio is defined as
\[
\mathcal{R}_S(w,r_f)=\frac{\textup{Reward}(w)-r_f}{\textup{Disp.}(w)}
\]
where $r_f$ is the risk free rate of return, and the tangency portfolio is the Sharpe ratio maximizing portfolio on the efficient frontier. In the NTS market model, the efficient frontier is not given by unique curve but surface, we have many tangency portfolios with respect to asymmetric tail risks. For each $\textup{Asym.}(w)$, we define the tangency portfolio as
\[
w_{tangency}(b) = \argmax_{w\in \mathcal{S}} \left\{\mathcal{R}_S(w,r_f) \Big| \textup{Asym.}(w)\ge b\right\},
\]
where $\mathcal{S} = \{w=(w_1,\cdots, w_N)\in I^N \,|\, I = [0,1], \sum_{n=1}^N w_n=1 \}$.

We define a new performance measure as
\[
\mathcal{R}_{AS}(b, r_f)=\frac{\mathcal{R}_S(w_{tangency}(b),r_f)}{A(b)},
\]
where
\[
A(b) = \frac{\sqrt\frac{2\theta}{2-\alpha}-b}{2\sqrt\frac{2\theta}{2-\alpha}}
\]
for the level $b$ of asymmetry tail risk.
The measure  $\mathcal{R}_{AS}(b, r_f)$ is referred to as the \textit{asymmetry scored Sharpe ratio} or simply \textit{AS ratio}. 
The value $A(b)$ increases when $b$ decreases, and $A(b)$ decreases when $b$ increases. Moreover, $0\le A(\textup{Asym.}(w)) \le 1$ since $|\textup{Asym.}(w)|\le\sqrt\frac{2\theta}{2-\alpha}$ if $w\in I^N$. Hence, we can use $A(\textup{Asym.}(w))$ as the evaluation index for a portfolio. Suppose we have two portfolios $P_w$ and $P_v$ having capital allocation weights $w$ and $v$, respectively. Assume that those two portfolio have the same Sharpe ratio. If $\text{Asym.}(w)>\text{Asym.}(v)$ i.e. $A(\textup{Asym.}(w))<A(\textup{Asym.}(v))$ then a risk averse investor may select the portfolio $P_w$ instead of $P_v$ because the return distribution of $P_v$ has the fatter negative tail than the return distribution of $P_w$. We look into this again in the following simple example.
\begin{example}\label{ex:Asym}
Suppose we have a market with three assets with the return $R=(R_1, R_2, R_3)^\tr$, and $R$ follows the NTS market model as
\[
R = \mu + \diag(\sigma) X, ~~~ X\sim \stdNTS_3(\alpha, \theta, \beta, \Sigma)
\]
where $\mu=(0.05, 0.05, 0.05)^\tr$, $\sigma=(\sqrt{0.08}, \sqrt{0.08}, \sqrt{0.08})^\tr$, $\alpha = 1.2$, $\theta = 1$, and $\beta = (1, 0, -1)^\tr$. In order to simplify the example we assume that $\cov(R_k, R_l)=0$ if $k\neq l$\footnote{For $k\neq l$, if the $(k,l)$-th element $\rho_{k,l}$ of $\Sigma$ is equal to 
\[
\rho_{k,l}=\frac{-\beta_k \beta_l\left(\frac{2-\alpha}{2\theta}\right)}{\sqrt{1-\beta_k\left(\frac{2-\alpha}{2\theta}\right)}\sqrt{1-\beta_l\left(\frac{2-\alpha}{2\theta}\right)}},
\]
then we have $\cov(R_k,R_l)=0$.}. Consider two portfolios $P_L$ and $P_D$ whose capital allocation weight vectors are $w_L=(0,0.5, 0.5)^\tr$ and $w_D=(0.5,0.5,0)^\tr$, respectively. Then, we have
\[
R_L = w_L^\tr R \disteq \bar\mu_L + \bar\sigma_L \Xi_L, \text{ and } R_D = w_D^\tr R \disteq \bar\mu_D + \bar\sigma_D \Xi_D,
\]
where $\bar\mu_L=\bar\mu_D=0.05$ and $\bar\sigma_L=\bar\sigma_D=0.2$. Since $\bar\beta(w_L)=\frac{0.5\sqrt{0.08}(-1)}{0.2}=-0.707$ and $\bar\beta(w_D)=0.707$, we have $\Xi_L\sim\stdNTS_1(1.2,1,-0.707, 1)$ and $\Xi_D\sim\stdNTS_1(1.2,1,0.707, 1)$. Hence, $P_L$ and $P_D$ have the same mean \& standard deviation, and the same Sharpe ratio. However, 
\[
A(\textup{Asym.}(w_L))=0.6581 > 0.3419 = A(\textup{Asym.}(w_D)),
\]
that means $P_L$ is more risky than $P_D$.
For that reason, A risk averse investor may select $P_D$ rather than $P_L$.
More detail illustration is presented in Figure \ref{fig:DvsL}. The dashed and solid curves are PDF of $R_L$ and $R_D$, respectively. The circle and inverted triangle are negative values of VaR for $P_L$ and $P_D$, respectively. VaR of $P_L$ is larger then VaR of $P_D$ at $1\%$ significant level. Therefore, $P_D$ is better than $P_L$ in view of the risk averse investor.
\end{example}
\begin{figure}[t]
\centering
\includegraphics[width = 12cm]{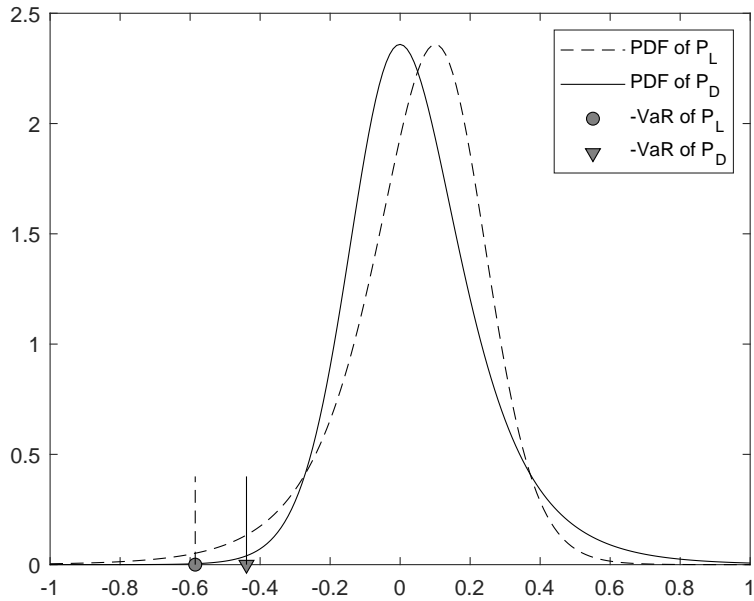}
\caption{\label{fig:DvsL} PDF of $P_L$ versus PDF of $P_D$. The dashed curve is PDF of $P_L$ and the solid curve is PDF od $P_D$. The circle and the inverted triangle are VaR's of $P_L$ and $P_D$, respectively, at 1\% sidnificant level.}
\end{figure}

Using the AS ratio, the Sharpe ratio of the tendency portfolio is evaluated by the level of the asymmetric tail risk. We find the AS ratio maximizing portfolio as follows:
\begin{align*}
P^* &= \max_{b\in[\min(\beta),\max(\beta)]} \mathcal{R}_{AS}(b, r_f)
\end{align*}
The $b^*$ corresponding $b$ value for the AS ratio maximizing portfolio $P^*$ is
\begin{align*}
b^* &= \argmax_{b\in[\min(\beta),\max(\beta)]} \mathcal{R}_{AS}(b, r_f).
\end{align*}
The capital allocation weight vector $w^*$ corresponding to $P^*$ is 
$w^* = w_{tangency}(b^*)$.

\begin{figure}
\centering
\includegraphics[width = 12cm]{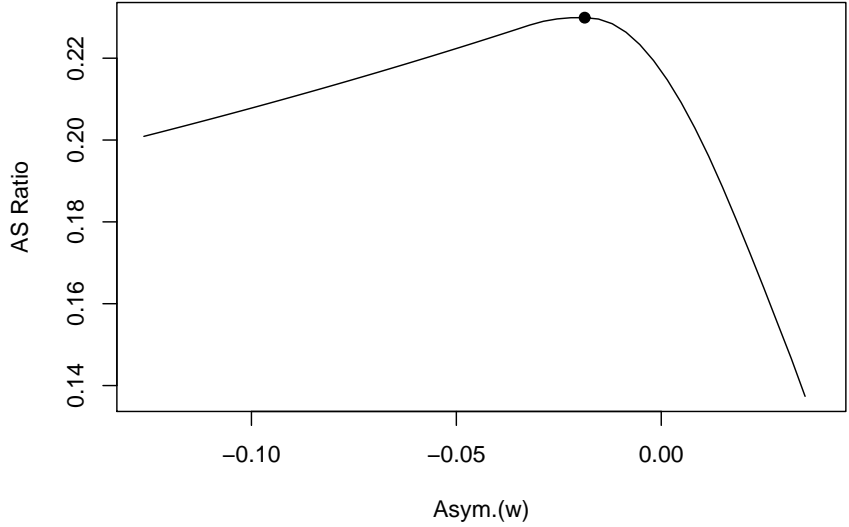}
\caption{\label{fig:perfMaxPoint}AS ratio curve. The $x$-axis is for the level $b$ of asymmetric tail risk and the $y$-axis is for AS ratios.}
\end{figure}
Figure \ref{fig:perfMaxPoint} presents the AS ratio curve with respect to the asymmetric tail risk based on $r_f=0.025/252$ and the estimated market parameters in Table \ref{table:ParamEst}. 
The solid curve is a $\mathcal{R}_{AS}(b, r_f)$ for $b\in[\min(\beta), \max(\beta)]$, and the dot point is the AS ratio maximizing portfolio $(b^*,P^*)=(-0.0187, 0.2299)$.\\
~\\
\begin{figure}[t]
\centering
\includegraphics[width=15cm]{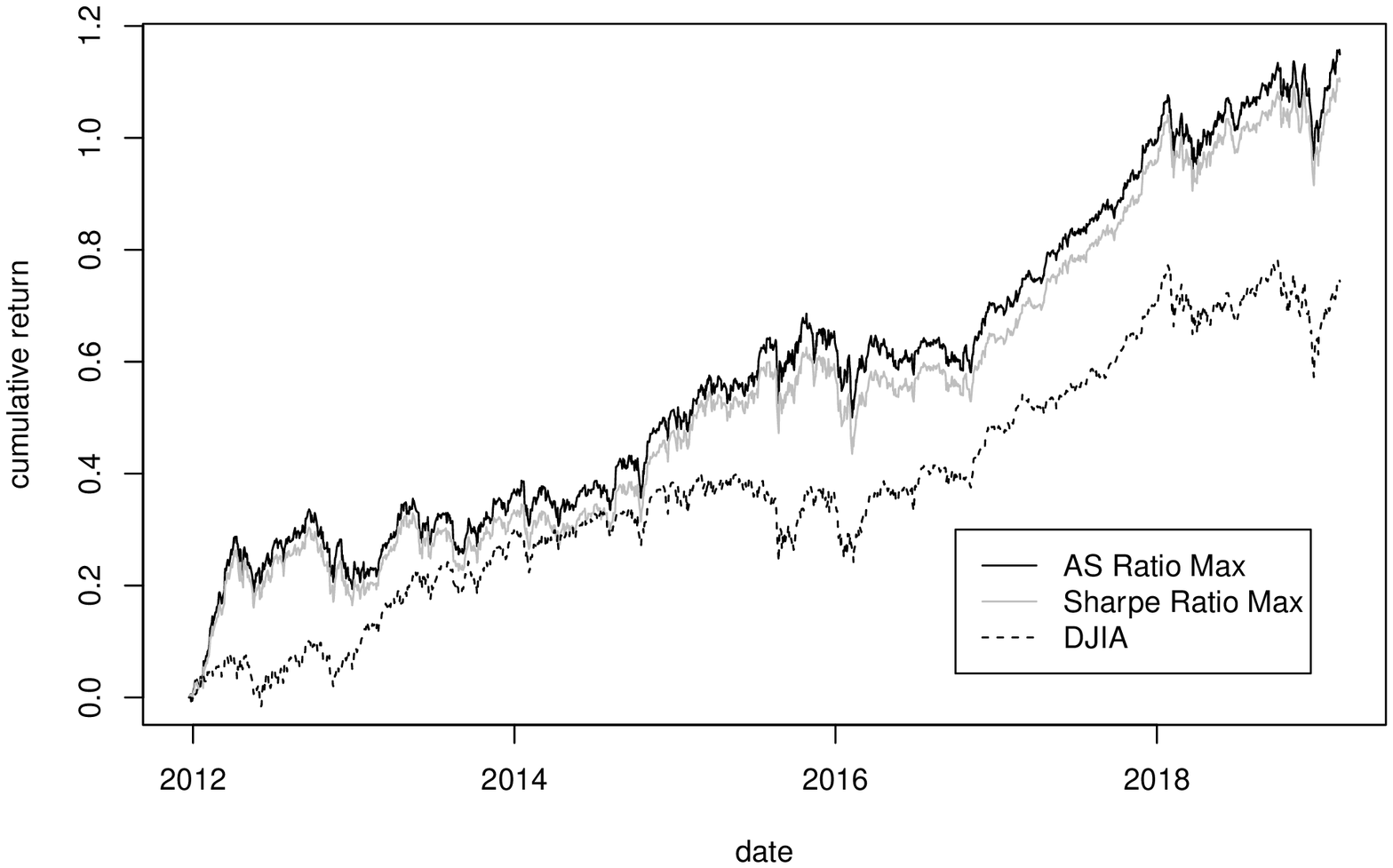}
\caption{\label{fig:hotseracing}Cumulative return curves for AS ratio maximization strategy (solid black), the Sharpe ratio maximization strategy (solid gray) and DJIA index (dashed black).}
\end{figure}
\textit{Trading Strategies}

We estimate parameters and rebalance portfolio every 10 days from 12/23/2011 to 02/21/2019.
We perform the parameter estimation explained in Section \ref{sec:NTS}, and then find the AS ratio maximizing portfolio and the Sharpe ratio maximizing portfolio for every 10 days. We set the risk free rate of return to be zero to simplify this investigation. We do not change those two portfolios for 10 days and re-balance them at the first day of the next 10 days. Figure \ref{fig:hotseracing} provides daily cumulative returns for those two strategies (the AS ratio maximizing and the Sharpe ratio maximizing strategies) and DJIA index as a bench mark portfolio return. In Table \ref{table:hotseracing}, we can see the performance of this investigation. The AS ratio maximization strategy has the larger mean return than the Sharpe ratio maximization strategy and DJIA index. The standard deviation, VaR and CVaR of the AS ratio maximization portfolio are larger than those values of the Sharpe ratio maximization portfolio. However, the Sharpe Ratio, VaR Ratio, and CVaR Ratio of the AS ratio maximization strategy are larger than those values of the Sharpe ratio maximization strategy.
The stdNTS parameters are also provided in the table for those three cases. The absolute value of beta parameter of DJIA is smallest among the three cases. The absolute value of beta for the AS ratio maximization strategy is smaller than that value of the Sharpe ratio maximization strategy. The AS ratio maximization strategy has largest $\mathcal{R}_{AS}(\beta, r_f)$ value comparing with those value of the Sharpe ratio maximization strategy and DJIA index. Therefore, the AS ratio maximization strategy has slightly better performance than the Sharpe ratio maximization strategy in this investigation. The AS ratio maximization strategy performs better than DJIA index, too.

\begin{table}[t]
\centering
\begin{tabular}{rrrrr}
  \hline
 & AS Ratio Max & Sharpe Ratio Max &  DJIA \\ 
  \hline
mean & $0.0645\%$ & $0.0618\%$ & $0.0419\%$ \\ 
  std dev & $0.984\%$ & $0.942\%$ & $0.804\%$ \\ 
  $\alpha$ & $0.46$ & $0.56$ & $0.39$ \\
  $\theta$ & $1.09$ & $0.96$ & $0.79$ \\
  $\beta$ & $-0.0603$ & $-0.0712$ & $-0.0525$\\
  \hline
  VaR & $2.49\%$ & $2.42\%$ &  $2.34\%$ \\ 
  CVaR & $3.22\%$ & $3.14\%$ &  $3.01\%$ \\ 
  \hline
  Sharpe Ratio & $6.61\%$ & $6.56\%$  & $5.21\%$ \\ 
  VaR Ratio & $2.59\%$ & $2.56\%$  & $1.79\%$ \\ 
  CVaR Ratio & $2.01\%$ & $1.97\%$ & $1.39\%$ \\ 
  $\mathcal{R}_{AS}(\beta, r_f)$ & $12.57\%$ & $12.35\%$ & $9.89\%$ \\
   \hline
\end{tabular}
\caption{\label{table:hotseracing}Performance of the AS ratio maximization strategy, the Sharpe ratio maximization portfolio and DJIA index. (``std dev" means standard deviation.)}
\end{table}

\section{Marginal Contribution to Risk}
In this section, we discuss the marginal contribution to risk on the NTS market model. We consider VaR and CVaR as risk measures, and find analytic solutions of the marginal contribution to VaR and to CVaR, respectively. Suppose $R$ is portfolio return as \eqref{eq:NTS Market} in the NTS market model. Let $w$ be a capital allocation rate vector, then $R_P(w)=w^\tr R$ is given by \eqref{eq:NTS Portfolio Return}. 

Let $F_\stdNTS(x,\alpha,\theta, \beta)$, $F_\stdNTS^{-1}(u,\alpha,\theta, \beta)$ and $\phi_\stdNTS(u, \alpha, \theta, \beta)$ be the CDF, the inverse CDF, and the characteristic function (Ch.F) of $\stdNTS_1(\alpha, \theta, \beta, 1)$. Then we have
\[
\VaR_\eta(\Xi) = -F_\stdNTS^{-1}(\eta,\alpha,\theta, \bar\beta(w)), ~~~\text{ and }~~~
\phi_\Xi(-u+i\delta) = \phi_\stdNTS(-u+i\delta, \alpha, \theta, \bar\beta(w)).
\]
Moreover, let $\CVaR_\stdNTS(\eta, \alpha, \theta, \beta)$ be the $\CVaR$ value of $\stdNTS_1(\alpha, \theta, \beta, 1)$ at the significant level $\eta$. If there is $\delta>0$ such that $|\phi_\Xi(-u+i\delta)|<\infty$ for all $u\in\R$
then we have\footnote{See Proposition 2 in \cite{Kim_et_al:2009b:AVaR} and \cite{Kim_et_al:2012}.}
\begin{align}
\nonumber
&\CVaR_\stdNTS(\eta, \alpha, \theta, \beta) \\
&= -F_\stdNTS^{-1}(u,\alpha,\theta, \beta) - \frac{1}{\pi\eta}\re\int_0^\infty e^{(iu+\delta)F_\stdNTS^{-1}(u,\alpha,\theta, \beta)}\frac{\phi_\stdNTS(-u+i\delta, \alpha, \theta, \beta)}{(-u+i\delta)^2}du,
\label{eq:CVaR stdNTS}
\end{align}
and
\[
\CVaR_\eta(\Xi)=\CVaR_\stdNTS(\eta, \alpha, \theta, \bar\beta(w)).
\]
Therefore, the VaR and CVaR of the portfolio return $R_P(w)$ are as follows:
\begin{align*}
\VaR_\eta(R_p(w)) &= -\bar\mu(w) + \bar\sigma(w) \VaR_\eta(\Xi)\\
&= -\bar\mu(w) - \bar\sigma(w) F_\stdNTS^{-1}(\eta, \alpha, \theta, \bar\beta(w))
\end{align*}
and
\begin{align*}
\CVaR_\eta(R_p(w)) &= -\bar\mu(w) +\bar\sigma(w) \CVaR_\eta(\Xi)\\
&=  -\bar\mu(w) +\bar\sigma(w)\CVaR_\stdNTS(\eta, \alpha, \theta, \bar\beta(w)).
\end{align*}

The first derivative of $\VaR_\eta(R_P(w))$ and the first derivative of $\CVaR_\eta(R_P(w))$ for $w_n$ are referred to as \textit{the marginal contribution to VaR} (MCT-VaR) and \textit{the marginal contribution to CVaR} (MCT-CVaR), respectively. They are provided the following proposition whose proof is in the Appendix.
\begin{proposition}\label{prop:MCT VaR and CVaR}
Suppose a portfolio return $R$ follows the NTS market model as \eqref{eq:NTS Market}. Let $w$ be a capital allocation rate vector, and $R_P(w)=w^\tr R$. The marginal contribution to VaR and to CVaR are equal to
\begin{align}
\nonumber
\frac{\partial \VaR_\eta(R_p(w))}{\partial w_n} 
&= -\mu_n -F_\stdNTS^{-1}(\eta, \alpha, \theta, \bar\beta(w))\frac{\partial \bar\sigma(w)}{\partial w_n} 
\\
&
 - \left(\sigma_n\beta_n - \bar\beta(w) \frac{\partial \bar\sigma(w)}{\partial w_n}\right)\frac{\partial F_\stdNTS^{-1}(\eta, \alpha, \theta, \beta)}{\partial \beta}\Big|_{\beta = \bar\beta(w)}
 \label{eq:MTC VaR}
\end{align}
and
\begin{align}
\nonumber \frac{\partial \CVaR_\eta(R_p(w))}{\partial w_n}
&
= -\mu_n + \CVaR_\stdNTS(\eta, \alpha, \theta, \bar\beta(w))\frac{\partial \bar\sigma(w)}{\partial w_n} 
\\
&
~~~ + \left(\sigma_n\beta_n - \bar\beta(w) \frac{\partial \bar\sigma(w)}{\partial w_n}\right)\frac{\partial \CVaR_\stdNTS(\eta, \alpha, \theta, \beta)}{\partial \beta}\Big|_{\beta = \bar\beta(w)},\label{eq:MCT CVaR}
\end{align}
respectively, where 
\begin{align*}
\frac{\partial \bar\sigma(w)}{\partial w_n} = \frac{\sum_{k=1}^Nw_k\cov(R_n, R_k)}{\bar\sigma(w)} ~~~\text{ for }~~~ n\in\{1,2,\cdots, N\},
\end{align*}
 and
\begin{align*}
\frac{\partial \CVaR_\stdNTS(\eta, \alpha, \theta, \beta)}{\partial \beta}
&=
-\frac{\partial F_\stdNTS^{-1}(\eta, \alpha, \theta, \beta)}{\partial \beta}
\\
&
-\frac{1}{\pi\eta}\re\int_0^\infty e^{(iu+\delta)F_\stdNTS^{-1}(\eta, \alpha, \theta, \beta)}\frac{\phi_\stdNTS(-u+i\delta, \alpha, \theta, \beta)}{(-u+i\delta)^2}
\\
&
~~~~~~\times\left((\delta+iu)\frac{\partial F_\stdNTS^{-1}(\eta, \alpha, \theta, \beta)}{\partial \beta}+\psi(-u+i\delta, \alpha, \theta, \beta)\right)du,
\end{align*}
with
\begin{align*}
&\psi(z, \alpha, \theta, \beta)
\\
&=-zi
	+\left(
		1-\frac{iz\beta}{\theta}
			+\left(
				1-\frac{\beta^2(2-\alpha)}{2\theta}
			\right)\frac{z^2}{2\theta}
	\right)^{\frac{\alpha}{2}-1}
	\left(
		zi
		+\frac{\beta(2-\alpha)}{2\theta}z^2
	\right).
\end{align*}
\end{proposition}

\subsection{Empirical Illustration}
In this subsection, we calculate MCT-VaR and MCT-CVaR of the NTS market model with the parameters fit to the empirical data and compare them to those values of a benchmark model. As a benchmark model, we consider the Gaussian market which is defined as
\[
R = \mu + \diag(\sigma) Z, ~~~ Z\sim \Phi(0, \Sigma_Z),
\]
where $\Phi(0, \Sigma_Z)$ is the $N$-dimensional standard normal distribution with covariance matrix $\Sigma_Z$ whose diagonal entries are 1. The portfolio return in this model is
\[
R_P = w^\tr R \disteq \bar\mu(w) + \bar\sigma(w)\epsilon,
\]
where $\epsilon\sim \Phi(0,1)$, $\bar\mu(w) = w^\tr \mu$, $\bar\sigma(w) = \sqrt{w^\tr \diag(\sigma)\Sigma_Z\diag(\sigma)w}$ and the capital allocation weight vector $w$.
Under this model, MCT-VaR and MCT-CVaR for significant level $\eta$ is equal to
\begin{equation}\label{eq:MCT-VaR Gaussian}
\frac{\partial \VaR_\eta(R_P(w))}{\partial w_n} = -\mu_n - F^{-1}_{\Phi(0,1)}(\eta)\frac{\sum_{k=1}^N w_k\cov(R_n, R_k)}{\bar\sigma(w)},
\end{equation}
and
\begin{equation}\label{eq:MCT-CVaR Gaussian}
\frac{\partial \CVaR_\eta(R_p(w))}{\partial w_n}
= -\mu_n + \CVaR_\eta(\epsilon)\frac{\sum_{k=1}^N w_k\cov(R_n, R_k)}{\bar\sigma(w)},
\end{equation}
where $F^{-1}_{\Phi(0,1)}$ is the inverse function of the CDF of $\Phi(0,1)$.
Comparing MCT-VaR \eqref{eq:MCT-VaR Gaussian} of the Gaussian model and MCT-VaR of \eqref{eq:MTC VaR} the NTS model, \eqref{eq:MTC VaR} has one more term related to $\beta$, than \eqref{eq:MCT-VaR Gaussian}. The asymmetry parameter $\beta$ determines skewness, and the left/right tail decaying of the distribution. We can say the same argument for MCT-CVaR \eqref{eq:MCT CVaR} and \eqref{eq:MCT-CVaR Gaussian}.

We consider equally weighted portfolio with 30 stocks in Table \ref{table:DJIA Members}, and use the estimated parameters in Table \ref{table:ParamEst}. The MCT-VaR and MCT-CVaR of the Gaussian market model are calculated by \eqref{eq:MCT-VaR Gaussian} and \eqref{eq:MCT-CVaR Gaussian} for $\eta=0.01$.
The MCT-VaR and MCT-CVaR of the NTS market model are calculated by Proposition \ref{prop:MCT VaR and CVaR} and results are presented in Table \ref{table:MCTR and Ranking}. The table exhibits MCT-VaR and MCT-CVaR for the Gaussian model and the NTS model together with the rank of the values for the ascending order.
The NTS market model describes skewness and fat-tails of portfolio return distribution, while the Gaussian market model is symmetric and its VaR and CVaR are effected by only dispersion factor. That generates the difference results between the Gaussian model column and the NTS model column in Table \ref{table:MCTR and Ranking}. For example, MCT-VaR and MCT-CVaR values for JPM rank 21th in the Gaussian model while those values rank 14th in the NTS model. The MCT-VaR and MCT-CVaR values of WBA rank 19th and 17th, respectively, in the Gaussian model, while they rank 23rd and 24th, respectively, in the NTS model.

\begin{table}
\centering
\begin{tabular}{c|cc|cc|cc|cc}
  \hline
  & \multicolumn{4}{c|}{Gaussian Model} & \multicolumn{4}{c}{NTS Model}\\
  \hline
  Symbol & \footnotesize{MCT-VaR(\%)} & Rank & \footnotesize{MCT-CVaR(\%)} & Rank  & \footnotesize{MCT-VaR(\%)} & Rank  & \footnotesize{MCT-CVaR(\%)}  &  Rank\\ 
  \hline
    \hline
AAPL & 2.23 &  24 & 2.58 &  24 & 2.91 &  22 & 3.99 &  22 \\ 
  AXP & 1.97 &  20 & 2.27 &  20 & 2.59 &  17 & 3.54 &  17 \\ 
  BA & 2.21 &  23 & 2.55 &  23 & 2.76 &  20 & 3.75 &  20 \\ 
  CAT & 2.77 &  30 & 3.18 &  30 & 3.59 &  30 & 4.89 &  30 \\ 
  CSCO & 2.36 &  26 & 2.71 &  26 & 3.27 &  29 & 4.49 &  29 \\ 
  CVX & 1.70 &  12 & 1.96 &  12 & 2.27 &  13 & 3.10 &  13 \\ 
  DD & 2.55 &  29 & 2.92 &  29 & 3.15 &  28 & 4.24 &  27 \\ 
  DIS & 1.53 &  10 & 1.76 &  10 & 1.85 &   7 & 2.50 &   7 \\ 
  GS & 2.38 &  27 & 2.73 &  27 & 3.00 &  24 & 4.05 &  23 \\ 
  HD & 1.73 &  14 & 1.99 &  14 & 2.38 &  15 & 3.28 &  15 \\ 
  IBM & 1.92 &  15 & 2.20 &  15 & 2.63 &  18 & 3.59 &  18 \\ 
  INTC & 2.43 &  28 & 2.79 &  28 & 3.09 &  26 & 4.20 &  26 \\ 
  JNJ & 1.34 &   7 & 1.54 &   7 & 1.88 &   8 & 2.60 &   8 \\ 
  JPM & 1.97 &  21 & 2.27 &  21 & 2.35 &  14 & 3.17 &  14 \\ 
  KO & 0.87 &   1 & 1.00 &   1 & 1.15 &   1 & 1.58 &   1 \\ 
  MCD & 0.96 &   3 & 1.11 &   3 & 1.30 &   3 & 1.80 &   3 \\ 
  MMM & 2.20 &  22 & 2.52 &  22 & 3.11 &  27 & 4.27 &  28 \\ 
  MRK & 1.25 &   6 & 1.45 &   5 & 1.60 &   5 & 2.18 &   5 \\ 
  MSFT & 2.27 &  25 & 2.62 &  25 & 3.01 &  25 & 4.13 &  25 \\ 
  NKE & 1.96 &  17 & 2.25 &  18 & 2.44 &  16 & 3.32 &  16 \\ 
  PFE & 1.45 &   8 & 1.67 &   9 & 2.02 &   9 & 2.77 &   9 \\ 
  PG & 0.95 &   2 & 1.10 &   2 & 1.21 &   2 & 1.65 &   2 \\ 
  RTX & 1.95 &  16 & 2.24 &  16 & 2.67 &  19 & 3.66 &  19 \\ 
  TRV & 1.45 &   9 & 1.67 &   8 & 2.11 &  11 & 2.90 &  11 \\ 
  UNH & 1.62 &  11 & 1.87 &  11 & 2.08 &  10 & 2.85 &  10 \\ 
  V & 1.96 &  18 & 2.26 &  19 & 2.78 &  21 & 3.86 &  21 \\ 
  VZ & 1.00 &   4 & 1.15 &   4 & 1.37 &   4 & 1.88 &   4 \\ 
  WBA & 1.96 &  19 & 2.25 &  17 & 2.95 &  23 & 4.07 &  24 \\ 
  WMT & 1.25 &   5 & 1.45 &   6 & 1.63 &   6 & 2.24 &   6 \\ 
  XOM & 1.72 &  13 & 1.97 &  13 & 2.26 &  12 & 3.06 &  12 \\  
   \hline   \hline
\end{tabular}
\caption{\label{table:MCTR and Ranking}MCT-VaR and MCT-CVaR for 30 stocks, under the Gaussian model and the NTS model. }
\end{table}

\subsection{Risk Budgeting}
Suppose a capital allocation weight vector $w$ is given and let $\varDelta w=(\varDelta w_1, \varDelta w_2, \cdots, \varDelta w_N)^\tr\in D$ where $D$ is a zero neighborhood in $\R^N$.
The optimal portfolios with respect to $\VaR$ and $\CVaR$ are obtained by solving the following problem:
\begin{align}
&\min_{\varDelta w} \varDelta \VaR_\eta (R_P(w))  \\
\nonumber
&\text{subject to }  \varDelta E[R_P(w)]\ge 0
\text{ and } \sum_{n=1}^N \varDelta w_n = 0,
\end{align}
and 
\begin{align}
&\min_{\varDelta w} \varDelta \CVaR_\eta (R_P(w))  \\
\nonumber
&\text{subject to }  \varDelta E[R_P(w)]\ge 0
\text{ and } \sum_{n=1}^N \varDelta w_n = 0.
\end{align}
where
\begin{align*}
&\varDelta E[R_P(w)] = E[R_p(w+\varDelta w)]-E[R_P(w)] \\
&\varDelta \VaR_\eta(R_p(w))= \VaR_\eta(R_p(w+\varDelta w))-\VaR_\eta(R_p(w))\\
&\varDelta \CVaR_\eta(R_p(w))= \CVaR_\eta(R_p(w+\varDelta w))-\CVaR_\eta(R_p(w)).
\end{align*}
Since we have
\begin{align*}
&\varDelta \VaR_\eta(R_p(w))\approx\sum_{n=1}^N \frac{\partial \VaR_\eta (R_p(w))}{\partial w_n}  \varDelta w_n,
\\
&\varDelta \CVaR_\eta(R_p(w))\approx\sum_{n=1}^N \frac{\partial \CVaR_\eta (R_p(w))}{\partial w_n}  \varDelta w_n.
\end{align*}
and 
\[
\varDelta E[R_P(w)]=\mu^\tr\varDelta w,
\]
we can find {\nop the} optimal portfolio on the local domain $D$ with respect to $\VaR$ and $\CVaR$, respectively, as follows:
\begin{align}\label{eq:OptPortVaR}
&\varDelta w^*=\argmin_{\varDelta w} \sum_{n=1}^N \frac{\partial \VaR_\eta (R_p(w)) }{\partial w_n} \varDelta w_n \\
\nonumber
&\text{subject to }  \mu^\tr \varDelta w\ge 0
\text{ and } \sum_{n=1}^N \varDelta w_n = 0.
\end{align}
and 
\begin{align}
\label{eq:OptPortAVaR}
&\varDelta w^*=\argmin_{\varDelta w} \sum_{n=1}^N \frac{\partial \CVaR_\eta (R_p(w))}{\partial w_n}  \varDelta w_n \\
\nonumber
&\text{subject to }  \mu^\tr \varDelta w\ge 0
\text{ and } \sum_{n=1}^N \varDelta w_n = 0.
\end{align}

\begin{table}
\centering
\begin{tabular}{crrr}
\hline
& Initial historical risk & Gaussian Model & NTS Model \\
 \hline
 $\VaR_{1\%}(R_P(w))$ & $2.6077\%$ & $2.5891\%$ & $2.5836\%$ \\
 $\varDelta\VaR_{1\%}(R_P(w))$ & & $-0.0186\%$ & $-0.0244\%$ \\
 \hline
 $\CVaR_{1\%}(R_P(w))$ &
	$3.3109\%$ & $3.3082\%$ & $3.3074\%$\\
 $\varDelta\CVaR_{1\%}(R_P(w))$ & & $-0.0027\%$ & $-0.0035\%$\\
\hline
\end{tabular}
\caption{\label{table:RiskBudgeting}Risk Budgeting Result. VaR and CVaR in this table are historical VaR and historical CVaR. }
\end{table}

We perform the risk budgeting for VaR and CVaR using the 30 stocks in Table \ref{table:DJIA Members} with the estimated parameters in Table \ref{table:ParamEst}.
Let the local domain be
\[
D=\{(x_1,x_2,\cdots,x_{30})\,|\, x_j\in[-3\cdot 10^{-3},3\cdot 10^{-3}],\, j=1,2,\cdots, 30 \},
\]
and the initial portfolio $w_0$ be the equally weighted portfolio. 
For risk budgeting for VaR under the NTS market model, we find $\varDelta w^*$ by \eqref{eq:OptPortVaR} with \eqref{eq:MTC VaR}.
For risk budgeting for CVaR, we find $\varDelta w^*$ by \eqref{eq:OptPortAVaR} with \eqref{eq:MCT CVaR}.
After that we put $w_{NTS}^{new}=w_0+\varDelta w^*$. 
We perform the risk budgeting under the Gaussian market model too for the benchmark. For risk budgeting for VaR, we find $\varDelta w^*$ by \eqref{eq:OptPortVaR} with \eqref{eq:MCT-VaR Gaussian}.
For risk budgeting for CVaR, we find $\varDelta w^*$ by \eqref{eq:OptPortAVaR} with \eqref{eq:MCT CVaR} under the NTS market model.
After that we put $w_{Gauss}^{new}=w_0+\varDelta w^*$.

Table \ref{table:RiskBudgeting} provides the results of this optimization. To compare the performance of the portfolio, we show the historical VaR and CVaR at 1\% significant level. Since the historical VaR and CVaR are model free, we fairly compare the risk budgeting performance of the Gaussian market model to that of the NTS market model. In VaR and CVaR, both Gaussian model and NTS model have negative increment after risk budgeting, but the $\varDelta\VaR$ and $\varDelta\CVaR$ in NTS model are smaller than those values in Gaussian model, respectively. Hence, the risk budgeting of the NTS model performs better than that of the Gaussian model.

Next, we perform the risk budgeting for VaR and CVaR, iteratively as the following algorithm:
\begin{enumerate}
\item[] \textbf{Step 1.} Consider the equally weighted portfolio $w$ as the initial portfolio.
\item[] \textbf{Step 2.} Calculate MCT-VaR or MCT-CVaR for $w$ 
\begin{itemize}
\item In risk budgeting for VaR, we use \eqref{eq:MTC VaR} for the NTS market model or use \eqref{eq:MCT-VaR Gaussian} for the Gaussian Model.
\item In risk budgeting for CVaR, we use \eqref{eq:MCT CVaR} for the NTS market model or use \eqref{eq:MCT-CVaR Gaussian} for the Gaussian Model.
\end{itemize}
\item[] \textbf{Step 3.} Perform risk budgeting and find $\varDelta w^*$
\begin{itemize}
\item In risk budgeting for VaR, we use \eqref{eq:OptPortVaR}.
\item In risk budgeting for CVaR, we use \eqref{eq:OptPortAVaR}.
\end{itemize}
\item[] \textbf{Step 4.} Change $w$ to $w+\varDelta w^*$ and go to Step 2. Repeat [Step 2 - Step 4] $M$ times.
\end{enumerate}  
Using the 30 stocks in Table \ref{table:DJIA Members} with the estimated parameters in Table \ref{table:ParamEst}, we perform the iterative risk budgeting $M=50$ times for 
\[
D=\{(x_1,x_2,\cdots,x_{30})| x_j\in\left[-2.5\cdot 10^{-4},2.5\cdot 10^{-4}\right], j=1,2,\cdots, 30 \}.
\]
The results are exhibited in Figure \ref{fig:RiskBudgetingIteration}. For each iteration, we calculate historical VaR, and CVaR. Those values are drawn in the first and the second plates, respectively. The figure shows that 
\begin{itemize}
\item VaR decreases in both the NTS Market model and the Gaussian market model, but the decreasing speed of the NTS model is faster than that of the Gaussian market model.
\item CVaR decreases in both the NTS Market model and the Gaussian market model, but the decreasing speed of the NTS model is faster than that of the Gaussian market model.
\end{itemize} 

Consequently, using risk budgeting of the NTS market model, we obtain the portfolio better performed than the portfolio obtained by the risk budgeting of the Gaussian market model. 
\begin{figure}
\centering
\includegraphics[width=13cm]{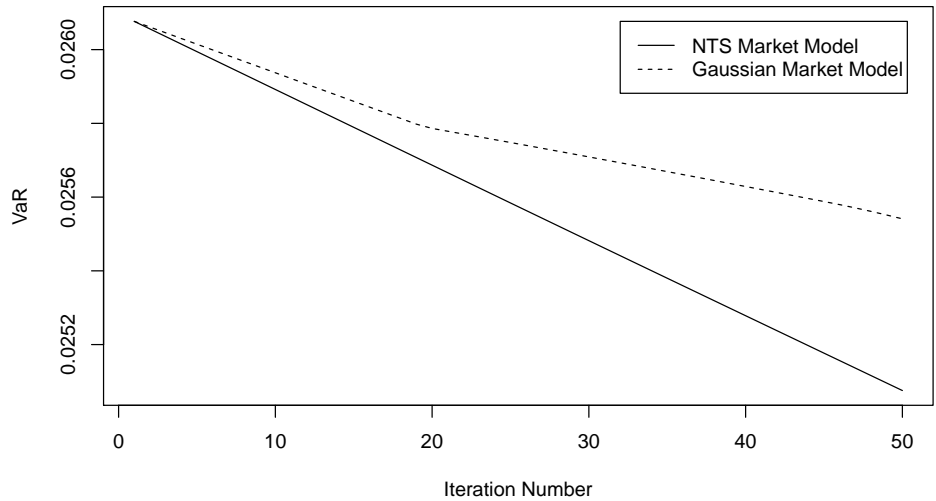}
\includegraphics[width=13cm]{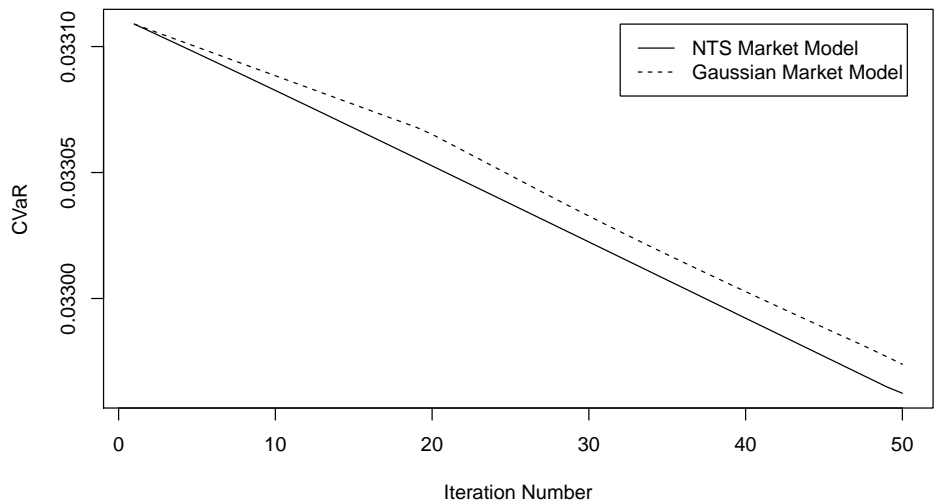}
\caption{\label{fig:RiskBudgetingIteration}Performance of risk budgeting. 
The first plate exhibit the VaR values for each iteration and the second plate exhibit CVaR values for each number of iteration. The solid lines are for the NTS market model and the dashed lines are for the Gaussian market model.
}
\end{figure}
\section{Conclusion}
In this paper, we discuss a generalized portfolio optimization considering not only mean and variance (dispersion) but also asymmetry. Using the NTS market model, we can decompose the portfolio risk to the dispersion risk and asymmetric risk, and obtain the generalized portfolio optimization method.
The classical efficient frontier for mean-variance optimization is extended to a surface on the three-dimensional space by the NTS market model. The Sharpe ratio and tangency portfolio can be extended to the AS ratio and the AS ratio maximization portfolio. The marginal contribution to risk on the NTS model is also given by analytically tractable forms for VaR and CVaR. 

Empirical tests are provided in the paper as well. The empirical surface of the efficient frontier on the NTS market model is presented for 30 major U.S. stock data, and AS ratio maximization strategy is discussed. Portfolio optimization methods by risk budgeting with the marginal contribution to VaR and CVaR are also performed on the NTS market model and the Gaussian market model. Using those empirical tests, we can see that the new portfolio optimization is more flexible and realistic than the traditional mean-variance method on the Gaussian model, and the risk budgeting on the NTS market model increases the performance of the portfolio faster than the risk budgeting on the Gaussian market model.

\appendix
\section{Appendix : Multivariate Normal Tempered Stable Distribution}
Let $\alpha\in(0,2)$ and $\theta>0$, and let $\Tau$ be a positive random variable whose characteristic function $\phi_{\Tau}$ is equal to
\begin{equation}\label{eq:ChF.TSsubordProcess}
\phi_{\Tau}(u) =
\exp\left(-\frac{2\theta^{1-\frac{\alpha}{2}}}{\alpha}\left((\theta-iu)^{\frac{\alpha}{2}}-\theta^{\frac{\alpha}{2}}\right)\right).
\end{equation}
The random variable $\Tau$ is referred to as \emph{Tempered Stable Subordinator}.
Let $X=(X_1, X_2, \cdots, X_N)^{\tr}$ be a multivariate random variable given by
\[
X = \mu + \beta(\Tau-1) + \textup{diag}(\gamma) \varepsilon \sqrt{\Tau} ,
\]
where 
\begin{itemize}
\item $\mu = (\mu_1, \mu_2, \cdots, \mu_N)^{\tr}\in\R^N$
\item $\beta = (\beta_1, \beta_2, \cdots, \beta_N)^{\tr}\in\R^N$
\item $\gamma = (\gamma_1, \gamma_2, \cdots, \gamma_N)^{\tr}\in\R_+^N$ with $\R_+=[0,\infty)$
\item $\varepsilon = (\varepsilon_1, \varepsilon_2, \cdots, \varepsilon_N)^{\tr}$ is a $N$-dimensional standard normal distribution with a covariance matrix $\Sigma$. That is, $\varepsilon_n\sim \Phi(0,1)$ for $n\in\{1,2,\cdots, N\}$ and $(k,l)$-th element of $\Sigma$ is given by $\rho_{k,l}=\cov(\varepsilon_k,\varepsilon_l)$ for $k,l\in\{1,2,\cdots,N\}$.
\item $\Tau$ is the Tempered Stable Subordinator with parameters $(\alpha,\theta)$, and is independent of $\varepsilon_n$ for all $n=1,2,\cdots, N$.
\end{itemize}
Then $X$ is referred to as the \emph{$N$-dimensional NTS random variable} with parameters $(\alpha$, $\theta$, $\beta$, $\gamma$, $\mu$, $\Sigma)$ which we denote by $X\sim \textup{NTS}_N(\alpha$, $\theta$, $\beta$, $\gamma$, $\mu$, $\Sigma)$.
The NTS distribution has the following properties:
\begin{enumerate}
\item The mean of $X$ are equal to $E[X] = \mu$.
\item The covariance between $X_k$ and $X_l$ is given by
\begin{equation}\label{eq:CovarianceAndRho}
\cov(X_k,X_l)=\rho_{k,l}\gamma_k\gamma_l+\beta_k\beta_l\left(\frac{2-\alpha}{2\theta}\right)
\end{equation}
for $k,l\in\{1,2,\cdots,N\}$. 
\item The variance of $X_n$ is 
\[
\var(X_n)=\gamma_n^2+\beta_n^2\left(\frac{2-\alpha}{2\theta}\right) \text{ for } n\in\{1,2,\cdots, N\}.
\]
\item Characteristic function of $X_n$ is
\[
\phi_{X_n}(u) = \exp\left((\mu-\beta)ui-\frac{2\theta^{1-\frac{\alpha}{2}}}{\alpha}
\left(\left(\theta-i\beta u+\frac{\gamma^2u^2}{2}\right)^{\frac{\alpha}{2}}-\theta^{\frac{\alpha}{2}}\right)\right) 
\]
\end{enumerate}

Providing $\mu_n=0$ and
$\gamma_n = \sqrt{1-\beta_n^2 \left(\frac{2-\alpha}{2\theta}\right)}$ with $|\beta_n|<\sqrt{ \frac{2\theta}{2-\alpha}}$ for $n$ $\in$ $\{ 1$,$2$, $\cdots$,$N\}$,
the $N$-dimensional NTS random variable $X$ has $E[X] = (0,0,\cdots,0)^{\tr}$ and $\var(X)$ $=$ $(1$,$1$,$\cdots$,$1)^{\tr}$. In this case, $X$ is referred to as the \emph{$N$-dimensional standard NTS random variable} with parameters $(\alpha$, $\theta$, $\beta$, $\Sigma)$ and we denote it by $X\sim \textup{stdNTS}_N(\alpha$, $\theta$, $\beta$, $\Sigma)$. 

For one dimensional NTS distribution, $\Sigma = 1$, we can prove the following Lemma which is changing parameterization.
\begin{lemma}\label{Lemma Scale NTS} Let $X\sim NTS_1(\alpha$, $\theta$, $\beta$, $\gamma$, $\mu$, $1)$ and $\xi\sim \textup{stdNTS}_1(\bar\alpha$, $\bar\theta$, $\bar\beta, 1)$.
Suppose 
$\bar\alpha = \alpha$, $\bar\theta = \theta$, and $
\bar\beta = \beta/\sigma$,
where $\sigma=\sqrt{\gamma^2+\beta^2\left(\frac{2-\alpha}{2\theta}\right)}$.
Then we have $X = \mu + \sigma \xi$.
\end{lemma}
\begin{proof}
The Ch.F of $X$ is given by
\begin{align}
\phi_{X}(u)&=\exp\left((\mu-\beta)iu-\frac{2\theta^{1-\frac{\alpha}{2}}}{\alpha}
\left(\left(\theta-i\beta u+\frac{\gamma^2 u^2}{2}\right)^{\frac{\alpha}{2}}-\theta^{\frac{\alpha}{2}}\right)\right)
\label{eq:chfXdt}
\end{align}
By the definition of stdNTS distribution, the Ch.F of $\mu+\sigma\xi$ is equal to
\begin{align}
\nonumber
\phi_{\mu+\sigma\xi}(u) &= E[e^{(\mu+\sigma\xi)ui}]=e^{\mu ui}E[e^{iu\sigma\xi}] \\
\label{eq:chf mu sigma xi}
&=\exp\left((\mu-\bar\beta \sigma) iu-\frac{2\bar\theta^{1-\frac{\bar\alpha}{2}}}{\bar\alpha}
\left(\left(\bar\theta-i\bar\beta \sigma u+\left(1-\bar\beta^2\left(\frac{2-\bar\alpha}{2\bar\theta}\right)\right)\frac{\sigma^2 u^2}{2}\right)^{\frac{\bar\alpha}{2}}-\bar\theta^{\frac{\bar\alpha}{2}}\right)\right).
\end{align}
Hence \eqref{eq:chfXdt}=\eqref{eq:chf mu sigma xi} if $\bar\alpha = \alpha$, $\bar\theta = \theta$, $\bar\beta \sigma = \beta$,  and $\gamma^2 =\sigma^2\left(1-\bar\beta^2\left(\frac{2-\alpha}{2\theta}\right)\right)$. Since $\sigma = \beta/\bar\beta$, we have
\begin{align*}
\gamma^2 = \frac{\beta^2}{\bar\beta^2}\left(1-\bar\beta^2\left(\frac{2-\alpha}{2\theta}\right)\right),
\end{align*}
or
\begin{align*}
\gamma^2\bar\beta^2 = \beta^2\left(1-\bar\beta^2\left(\frac{2-\alpha}{2\theta}\right)\right)= \beta^2-\beta^2\bar\beta^2\left(\frac{2-\alpha}{2\theta}\right),
\end{align*}
and hence
\begin{align*}
\bar\beta^2\left(\gamma^2+\beta^2\left(\frac{2-\alpha}{2\theta}\right)\right)=\beta^2.
\end{align*}
Therefore, we have
\begin{equation*}
\bar\beta = \frac{\beta}{\sigma},
\end{equation*}
where 
\[
\sigma = \sqrt{\gamma^2+\beta^2\left(\frac{2-\alpha}{2\theta}\right)}.
\]
\end{proof}

The linear combination of NTS member variables of the NTS vector is again NTS distributed as the following proposition.
\begin{lemma} \label{pro:LinCombMNTS}
Let $w = (w_1, w_2, \cdots, w_N)^\tr\in\R^N$ and $X\sim \textup{NTS}_N(\alpha$, $\theta$, $\beta$, $\gamma$, $\mu$, $\Sigma)$. Then 
$w^{\tr} X \sim \textup{NTS}_1(\alpha,\theta,\bar\beta,\bar\gamma,\bar\mu,1)$,
where
\begin{align*}
\bar\mu=w^{\tr} \mu, ~~~\bar\beta=w^{\tr} \beta~~~
\text{ and } ~~~
\bar\gamma=\sqrt{w^\tr \diag(\gamma)\Sigma\diag(\gamma) w }.
\end{align*}
\end{lemma}
\begin{proof}
Since we have
\begin{align*}
w^\tr X = w^\tr \mu + w^\tr \beta (\Tau-1) + w^\tr \diag(\gamma)\epsilon \sqrt{\Tau},
\end{align*}
and $w^\tr \diag(\gamma)\epsilon \disteq \sqrt{w^\tr \diag(\gamma)\Sigma\diag(\gamma) w } \epsilon_0$ with $\epsilon_0\sim \Phi(0,1)$, it is trivial.
\end{proof}
Finally we can provide the proof of Proposition \ref{pro:mu+sigma stdNTS}.
\subsection{Proof of Proposition \ref{pro:mu+sigma stdNTS}}
\begin{proof}[Proof of Proposition \ref{pro:mu+sigma stdNTS}]
By \eqref{eq:NTS Market} and \eqref{eq:NTS Portfolio Return}, we have
\begin{align*}
R &= \mu + \diag(\sigma)(\beta(\Tau-1)+\diag(\gamma)\epsilon \sqrt{\Tau})\\
&= \mu + \diag(\sigma)\beta(\Tau-1)+\diag(\sigma)\diag(\gamma)\epsilon \sqrt{\Tau}
\\&
\sim \textup{NTS}_N (\alpha, \theta, \diag(\sigma) \beta, \diag(\sigma) \gamma, \mu, \Sigma ),
\end{align*}
where $\gamma=(\gamma_1, \gamma_2, \cdots, \gamma_N)^{\tr}$ with $\gamma_n = \sqrt{1-\beta_n^2 \left(\frac{2-\alpha}{2\theta}\right)}$ and $\Tau$ is the tempered stable subordinator with parameter $(\alpha,\theta)$.
By Lemma \ref{pro:LinCombMNTS}, 
\[
R_P(w) = w^\tr R \sim \NTS_1(\alpha, \theta, w^\tr\diag(\sigma)\beta, \sqrt{w^\tr \diag(\sigma)\diag(\gamma) \Sigma \diag(\gamma)\diag(\sigma)w}, w^\tr\mu, 1).
\]
By Lemma \ref{Lemma Scale NTS}, we have
\[
R_P(w) = w^\tr\mu + \bar\sigma(w) \Xi
\]
where
\[
\bar\sigma(w) = \sqrt{w^\tr \diag(\sigma)\diag(\gamma) \Sigma \diag(\gamma)\diag(\sigma)w+(w^\tr\diag(\sigma)\beta)^2\left(\frac{2-\alpha}{2\theta}\right)}
\]
and
\[
\Xi\sim\stdNTS_1(\alpha, \theta, w^\tr\diag(\sigma)\beta/\bar\sigma(w), 1).
\]
Also, by \eqref{eq:CovarianceAndRho}, we have
\begin{align*}
\left(\bar\sigma(w)\right)^2 &= \sum_{k=1}^N\sum_{l=1}^N w_k w_l \gamma_k\gamma_l\sigma_k\sigma_l \rho_{k,l} + \left(\sum_{k=1}^N w_k\sigma_k\beta_k\right)^2\left(\frac{2-\alpha}{2\theta}\right)\\
&= \sum_{k=1}^N\sum_{l=1}^N w_k w_l \sigma_k\sigma_l \left(\gamma_k\gamma_l\rho_{k,l}+\beta_k\beta_l\left(\frac{2-\alpha}{2\theta}\right)\right)
\\
&
= \sum_{k=1}^N\sum_{l=1}^N w_k w_l \sigma_k\sigma_l\cov(X_k,X_l).
\end{align*}
Hence, we have
\[
\left(\bar\sigma(w)\right)^2 = \sum_{k=1}^N\sum_{l=1}^N w_k w_l \cov(R_k, R_l) = w^\tr \Sigma_R w
\]
where $\Sigma_R$ is the covariance matrix of $R$.
\end{proof}
\subsection{Proof of Proposition \ref{prop:MCT VaR and CVaR}}
\begin{proof}[Proof of Proposition \ref{prop:MCT VaR and CVaR}]
By \eqref{eq:NTS Portfolio Return}, we have
\begin{equation*} 
R_P(w) = \bar\mu(w) + \bar\sigma(w) \Xi ~~~\text{ for }~~~ \Xi\sim \textup{stdNTS}_1(\alpha, \theta, \bar\beta(w), 1),
\end{equation*}
where 
\begin{align*}
\bar\mu(w) &= w^\tr \mu = \sum_{k=1}^N w_k\mu_k, \\
\bar\beta(w) &= \frac{w^\tr\diag(\sigma)\beta}{\bar\sigma(w)}= \frac{\sum_{k=1}^N w_k\sigma_k\beta_k}{\bar\sigma(w)},\\
\bar\sigma(w)&=\sqrt{w^\tr \Sigma_R w}=\sqrt{\sum_{k=1}^N\sum_{l=1}^N w_k w_l \cov(R_k, R_l)}.
\end{align*}
Hence, the first derivative of $\bar\beta(w)$ and $\bar\sigma(w)$ are obtained as follows:
\begin{align}
\nonumber
\frac{\partial \bar\beta(w)}{\partial w_n} &= \frac{\partial}{\partial w_n}
\left(\frac{\sum_{k=1}^N w_k \sigma_k \beta_k}{\bar\sigma(w)}\right)
\\
\nonumber
&=\frac{\sigma_n\beta_n\bar\sigma(w) - \left(\sum_{k=1}^N w_n \sigma_k \beta_k\right)\frac{\partial \bar\sigma(w)}{\partial w_n}}{\left(\bar\sigma(w)\right)^2}
\\
&=\frac{\sigma_n\beta_n}{\bar\sigma(w)} - \frac{\bar\beta(w)}{\bar\sigma(w)} \frac{\partial \bar\sigma(w)}{\partial w_n}
\label{eq:d beta dw}
\end{align}
and
\begin{align*}
\frac{\partial \bar\sigma(w)}{\partial w_n} & = \frac{1}{2\bar\sigma(w)}
\frac{\partial }{\partial w_n}\sum_{k=1}^N\sum_{l=1}^N w_k w_l \cov(R_k, R_l)
\\
& 
= \frac{\sum_{k=1}^Nw_k\cov(R_n, R_k)}{\bar\sigma(w)}.
\end{align*}
Since we have
\[
\VaR_\eta(R_p(w)) = -\bar\mu(w) - \bar\sigma(w) F_\stdNTS^{-1}(\eta, \alpha, \theta, \bar\beta(w))
\]
we obtain
\begin{align}
\nonumber
&\frac{\partial \VaR_\eta(R_p(w))}{\partial w_n} 
\\
\nonumber
&= -\frac{\partial \bar\mu(w)}{\partial w_n} - \frac{\partial \bar\sigma(w)}{\partial w_n} F_\stdNTS^{-1}(\eta, \alpha, \theta, \bar\beta(w)) - \bar\sigma(w)\frac{\partial F_\stdNTS^{-1}(\eta, \alpha, \theta, \bar\beta(w))}{\partial w_n}
\\&
=-\mu_n - \frac{\partial \bar\sigma(w)}{\partial w_n}F_\stdNTS^{-1}(\eta, \alpha, \theta, \bar\beta(w))
 + \bar\sigma(w)\frac{\partial F_\stdNTS^{-1}(\eta, \alpha, \theta, \beta)}{\partial \beta}\Big|_{\beta = \bar\beta(w)}\frac{\partial \bar\beta(w)}{\partial w_n} \label{eq:d Var Xi}
\end{align}
By substituting \eqref{eq:d beta dw} into \eqref{eq:d Var Xi}, we obtain \eqref{eq:MTC VaR}.

Since there is $\delta>0$ such that $|\phi_\Xi(-u+i\delta)|<\infty$ for all $u\in\R$, we have 
\begin{align*}
\frac{\partial \CVaR_\eta(\Xi)}{\partial \beta}\Big|_{\beta=\bar\beta(w)}=\frac{\partial }{\partial \beta}\CVaR_\stdNTS(\eta, \alpha, \theta, \beta)\Big|_{\beta=\bar\beta(w)}.
\end{align*}
By \eqref{eq:CVaR stdNTS}, we have
\begin{align*}
&
\frac{\partial }{\partial \beta}\CVaR_\stdNTS(\eta, \alpha, \theta, \beta)
\\
&
=-\frac{\partial}{\partial \beta}F_\stdNTS^{-1}(u,\alpha,\theta, \beta)
\\
&
-\frac{1}{\pi\eta}\re\int_0^\infty e^{(iu+\delta)F_\stdNTS^{-1}(u,\alpha,\theta, \beta)}\frac{\phi_\stdNTS(-u+i\delta, \alpha, \theta, \beta)}{(-u+i\delta)^2}\\
&\times\left((\delta+iu)\frac{\partial }{\partial \beta}F_\stdNTS^{-1}(u,\alpha,\theta, \beta) +\frac{\partial}{\partial \beta}\log \phi_\stdNTS(-u+i\delta, \alpha, \theta, \beta)\right)du.
\end{align*}
By setting $\psi(z, \alpha, \theta, \beta)=\frac{\partial}{\partial \beta}\log\phi_\stdNTS(z, \alpha, \theta, \beta)$, we can simplify
\begin{align*}
\frac{\partial }{\partial \beta}\CVaR_\stdNTS(\eta, \alpha, \theta, \beta)
&
=-\frac{\partial}{\partial \beta}F_\stdNTS^{-1}(u,\alpha,\theta, \beta)
\\
&
-\frac{1}{\pi\eta}\re\int_0^\infty e^{(iu+\delta)F_\stdNTS^{-1}(u,\alpha,\theta, \beta)}\frac{\phi_\stdNTS(-u+i\delta, \alpha, \theta, \beta)}{(-u+i\delta)^2}\\
&\times\left((\delta+iu)\frac{\partial }{\partial \beta}F_\stdNTS^{-1}(u,\alpha,\theta, \beta) +\psi(-u+i\delta, \alpha, \theta, \beta)\right)du.
\end{align*}
The characteristic function  $\phi_\stdNTS(u, \alpha, \theta, \beta)$ is equal to
\begin{align*}
\nonumber
&\phi_\stdNTS(u, \alpha, \theta, \beta) \\
&=\exp\left(-\beta ui-\frac{2\theta^{1-\frac{\alpha}{2}}}{\alpha}
\left(\left(\theta-i\beta u+\left(1-\frac{\beta^2(2-\alpha)}{2\theta}\right)\frac{u^2}{2}\right)^{\frac{\alpha}{2}}-\theta^{\frac{\alpha}{2}}\right)\right),
\end{align*}
hence we have
\begin{align*}
&\psi(z, \alpha, \theta, \beta)
\\&
=\frac{\partial}{\partial \beta}
\left(-\beta zi-\frac{2\theta^{1-\frac{\alpha}{2}}}{\alpha}
\left(\left(\theta-i\beta z+\left(1-\frac{\beta^2(2-\alpha)}{2\theta}\right)\frac{z^2}{2}\right)^{\frac{\alpha}{2}}-\theta^{\frac{\alpha}{2}}\right)\right)
\\
&
=-zi
	+\left(
		1-\frac{iz\beta}{\theta}
			+\left(
				1-\frac{\beta^2(2-\alpha)}{2\theta}
			\right)\frac{z^2}{2\theta}
	\right)^{\frac{\alpha}{2}-1}
	\left(
		zi
		+\frac{\beta(2-\alpha)}{2\theta}z^2
	\right).
\end{align*}
As VaR case, CVaR for $R_P(w)$ is calculated using $\CVaR\eta(\Xi)$ that
\begin{align*}
\CVaR_\eta(R_p(w)) = -\bar\mu(w) +\bar\sigma(w) \CVaR_\eta(\Xi).
\end{align*}
Therefore, we have
\begin{align}
\nonumber &
\frac{\partial \CVaR_\eta(R_p(w))}{\partial w_n} 
\\
\nonumber 
&= -\mu_n + \frac{\partial \bar\sigma(w)}{\partial w_n} \CVaR_\eta(\Xi) + \bar\sigma(w)\frac{\partial\CVaR_\eta(\Xi)}{\partial w_n} 
\\
\nonumber 
&= -\mu_n + \frac{\partial \bar\sigma(w)}{\partial w_n} \CVaR_\eta(\Xi) + \bar\sigma(w)\frac{\partial\CVaR_\eta(\Xi)}{\partial \beta}\Big|_{\beta=\bar\beta(w)}\frac{\partial \bar\beta(w)}{\partial w_n}\\
&= -\mu_n + \frac{\partial \bar\sigma(w)}{\partial w_n} \CVaR_\eta(\Xi) + \left(\sigma_n\beta_n - \bar\beta(w) \frac{\partial \bar\sigma(w)}{\partial w_n}\right)\frac{\partial\CVaR_\eta(\Xi)}{\partial \beta}\Big|_{\beta = \bar\beta(w)} \label{eq:d CVaR Xi}
\end{align}
By substituting $\CVaR_\eta(\Xi)=\CVaR_\stdNTS(\eta,\alpha,\theta,\bar\beta(w))$ into \eqref{eq:d CVaR Xi}, we obtain \eqref{eq:MCT CVaR}.
\end{proof}

\singlespacing
\bibliographystyle{decsci_mod}
\bibliography{refs_aaron_EffNTS}

\end{document}